\documentclass[a4paper,11pt]{article}
\usepackage[margin=0.83in]{geometry}
\usepackage{amsthm,amsmath}
\usepackage{amssymb}

\usepackage{times}
\usepackage[ruled]{algorithm}
\usepackage[noend]{algpseudocode}

\usepackage{epsfig}

\newtheorem{theorem}{Theorem}

\newtheorem{lemma}{Lemma}
\newtheorem{claim}{Claim}
\newtheorem{fact}{Fact}

\newtheorem{proposition}{Proposition}

\newtheoremstyle{redstyle}% name
     {3pt}%      Space above
     {3pt}%      Space below
     {\color{black}}%         Body font
     {}%         Indent amount (empty = no indent, \parindent = para indent)
     {\color{red}\bfseries}% Thm head font
     {:}%        Punctuation after thm head
     {.5em}%     Space after thm head: " " = normal interword space;
     {}%         Thm head spec (can be left empty, meaning `normal')
\theoremstyle{redstyle}

\newcommand{\comment}[1]{}

\usepackage{color}

\newcommand{\labell}[1]{\label{#1}} %\marginpar{#1}}
\newcommand{\todo}[1]{\noindent\colorbox{red}{todo: #1}}

\newcommand{\cc}[1]{{\color{cyan}{#1}}}
\comment{
\newcommand{\labell}[1]{\label{#1}}
\newcommand{\todo}[1]{}

\newcommand{\cc}[1]{}
} %end comment
\newcommand{\NAT}{{\mathbb N}}

\newcommand{\dist}{\text{dist}}

\newcommand{\cT}{{\mathcal T}}
\newcommand{\cN}{{\mathcal N}}

\newcommand{\cI}{{\mathcal I}}

\newcommand{\eps}{\varepsilon}

\newfont{\mycrnotice}{ptmr8t at 7pt}
\newfont{\myconfname}{ptmri8t at 7pt}

\begin{document}

\title{On the Impact of Geometry on Ad Hoc Communication \\ in Wireless Networks
\thanks{%
	This work was supported by the Polish National Science Centre grant
DEC-2012/07/B/ST6/01534.
}
}

\author{Tomasz Jurdzinski$^1$ \and Dariusz R.~Kowalski$^2$ \and Michal Rozanski$^1$ \and Grzegorz Stachowiak$^1$}

\footnotetext[1]{Institute of Computer Science, University of Wroc{\l}aw, Poland.}

\footnotetext[2]{Department of Computer Science,
            University of Liverpool,
            Liverpool L69 3BX, UK.
            }

\date{}

\maketitle

\begin{abstract}
In this work we address the question how important is the knowledge of geometric location and network density to the efficiency of (distributed) wireless communication in ad hoc networks. We study fundamental communication task of broadcast and develop well-scalable, randomized algorithms that do not rely on GPS information, and which efficiency formulas do not depend on how dense the geometric network is. We consider two settings: with and without spontaneous wake-up of nodes. In the former setting, in which all nodes start the protocol at the same time, our algorithm accomplishes broadcast in $O(D\log n + \log^2 n)$ rounds under the SINR model, with high probability (whp), where $D$ is the diameter of the communication graph and $n$ is the number of stations.
In the latter setting, in which only the source node containing the original message is active in the beginning, we develop a slightly slower algorithm working in $O(D\log^2 n)$ rounds whp. Both algorithms are based on a novel distributed coloring method, which is of independent interest and potential applicability to other communication tasks under the SINR wireless model.
\vspace*{1ex}

\end{abstract}

\section{Introduction}

In this paper we study distributed communication problems in wireless networks,
where interferences are resolved by the Signal-to-Interference-and-Noise Ratio (SINR)
physical model.
Specifically, but not exclusively, we concentrate on the broadcast problem,
where a piece of information stored in a specified station/node (the source) is supposed
%broadcast its message
to be delivered
to all other stations in the network.
The broadcast is a
fundamental communication primitive, whose complexity is well understood
in the previous models of wireless communication, such as
%graph-based
radio networks.
%model of wireless networks.
Closer to reality, models based on the SINR constraint
%model
attracted attention of algorithmic community much later
than the radio network model. One of the key differences between models is that
radio networks take into account only interference between stations in close neighborhood,
while the SINR model relies on the physical assumptions that the strength of
signals decrease gradually according to a continuous function and cumulate, which makes
development of algorithms and their analysis much more complicated.
\comment{
In order to address this challenge, the authors studying the broadcast problem
in the SINR model often assume various helpful capabilities.
%take advantage of several features specific for the model
%which help to design efficient algorithms.
%In particular,
For instance,
they use a feature of tunable collision
detection, assume that all stations start a protocol simultaneously (which allows for a preprocessing),
power control allowing stations to decide
the strength of a transmitted signal in each step %(e.g., Yu et al.),
assumption that all stations know their positions in the Euclidean space, or
knowledge of granularity (maximum ratio between
between actual distances of stations which can communicate directly???) of a network.
Usually, the knowledge of parameters of the SINR model is also assumed.

The goal of this paper is to develop algorithms applicable in networks
of severely limited devices (sensors).
In particular, we assume that most of the features listed above are not
available.
In several applications, sensors communicating by wireless medium are
severely limited which naturally motivates such approach.
More generally, if efficient algorithms in such a model are possible, this
shows which features are \textbf{not} necessary for efficient communication.
}

In this work we show that efficiency of wireless communication depends
%only
mainly
on
parameters of the communication graph, even for devices with limited knowledge
and capabilities. In particular, we do not assume any carrier sensing capabilities,
initial synchronization
or any knowledge other than rough estimates of the number of nodes and physical SINR parameters.
Despite of that, we develop almost optimal and well scalable solutions to the broadcast
and wake-up problems. %, with application to other tasks.
Moreover, as mentioned above, the worst-case performance
of the considered communication tasks, and other problems building on them,
depends only on the topology of the communication graph, also called a reachability graph,
and not on specific location of nodes within reachability balls.
One of the implications is that geometric properties of reachability regions
studied in some previous works, c.f.,~\cite{Peleg-JACM}, do not influence worst-case
scenarios in ad hoc communication by more than $O(\log^2 n)$ factor - the factor
by which our algorithms are far from lower bounds.

\subsection{Model}

We consider the model of a wireless network consisting of {\em stations}, also called {\em nodes},
deployed into a metric space with {\em bounded growth property} of degree $\gamma$.\footnote{%
This notion generalizes the Euclidean $R^\gamma$ space; its formal definition is provided later in this section.}
%Each station $v$ has its {\em transmission power} $P_v$, which is a positive real number.
All stations are identical, and therefore each of them has the same transmission power $P$
(we call it a {\em uniform power model}).

There are three fixed model parameters: path loss
$\alpha> \gamma$,
threshold
%$\beta> 1$,
$\beta\ge 1$,
ambient noise
$\cN>0$.
We also assume a connectivity graph parameter $\eps\in (0,1)$.

The $SINR(v,u,\cT)$ ratio, for given stations $u,v$ and a set of (transmitting) stations $\cT$,
is defined as follows:
\vspace*{-0.5ex}
\begin{equation}\label{e:sinr}
SINR(v,u,\cT)
=
\frac{P\dist(v,u)^{-\alpha}}{\cN+\sum_{w\in\cT\setminus\{v\}}P_w\dist(w,u)^{-\alpha}}
\end{equation}
In the {\em Signal-to-Interference-and-Noise-Ratio (SINR) model}
%considered in this work,
a station $u$ successfully receives a message from a station $v$ in a round if
$v\in \cT$, $u\notin \cT$ and $SINR(v,u,\cT)\ge\beta \ ,$ where $\cT$ is the set of stations transmitting in that round.

\paragraph{Synchronization}
It is assumed that algorithms work synchronously in rounds. In general, we do not assume global clock ticking.
Note, however, that some kind of global synchronization can be achieved by appending a counter to every message sent throughout broadcast algorithm; we use this property in algorithms developed in this work.

\paragraph{Carrier sensing}
We consider the model {\em without carrier sensing}, that is,
a station $u$ has no other feedback from the wireless channel than
receiving or not receiving a message in a round $t$.

\paragraph{Knowledge of stations}
Each station knows the number of stations in the network, $n$.
Our algorithms also work when stations share, instead of $n$, an estimate $\nu \ge n$ of this value which is $O(n^c)$  for a fixed constant $c$.
%We assume a simple store-and-forward model, i.e.,
%(Thus, in particular, a station is not able
%to determine whether two received messages where transmitted by the same station.)
%
We assume that nodes do not know the precise value of the SINR parameters
$\alpha$, $\beta$, and $\cN$ but instead know only upper and lower bounds for
the parameters (i.e., $\alpha_{\min}$ and $\alpha_{\max}$, $\beta_{\min}$ and $\beta_{\max}$,
$\cN_{\min}$ and $\cN_{\max}$).
For simplicity, in this version of the paper we perform calculations assuming
that exact values of these parameters are known.
In order to take into account uncertainty regarding those parameters,
it is sufficient to choose their maximal/minimal values depending on
the fact whether upper or lower estimates are provided.
%cala techniczna czesc jest na razie tak napisana, jakby te parametry byly
%znane; sprawdzic ktore wartosci wystarczy przyjac w rachunkach; gdzie maksymalne i gdzie
%minimalne; i czy taki zabieg wystarczy. czy parametry musza byc takie same dla wszystkich "linkow"?

\paragraph{Messages and initialization of stations} % other than source
We consider two variants of initialization of stations: without spontaneous wake-up and with spontaneous wake-up.
In the former model each station (except some distinguished one(s)) sleeps till it obtains the %broadcasted 
message for the first time.
In the latter variant, all nodes are woken up 
at the same time and start an execution of an algorithm simultaneously.
%before broadcast is performed --
%this enables some preprocessing done by the stations not possessing the message that speeds %up the actual broadcast.}
(Observe that nodes can benefit from the spontaneous wake-up setting by performing a local preprocessing
%to organize
simultaneously for the whole network.)
%, before messages with actual content
%would arrive.)

%A station other than the source starts executing the broadcast protocol
%after the first successful receipt of a message;
%the source message;
%it is often called
%a {\em non-spontaneous wake-up model}. 
Each station can
either act as a sender or as a receiver during a round.
A sender can transmit a broadcast message
with attaching to it $O(\log n)$ additional bits.
%
%We say that a station which receives the source message for the first
%time is {\em waken up} at this moment and it is awake afterwards.
Our algorithms are described from
a ``global'' perspective, i.e., we count rounds starting from the moment
when the first message is sent.
%source sends its first message.
In order to synchronize stations in the model with non-spontaneous wake-up, we assume
that each message contains the number of rounds elapsed from the beginning
of the execution of the algorithm.
%\footnote{na potrzeby dolnej granicy mozemy zalozyc,
%ze nie mozna przeslac nic wiecej niz stan zegara}
%
%\dk{ponizsze chyba zbedne, lub wymaga integracji z powyzszym}
%
%\tj{We consider also the {\em spontaneous wake-up model},
%where all stations start a protocol's execution at the same moment.
%(Observe that nodes can benefit from the spontaneous wake-up setting by performing a local %preprocessing
%to organize
%simultaneously for the whole network, before messages with actual content
%would arrive.)
%simultaneously for the following broadcast of the source message.

\paragraph{Ranges and uniformity}
The {\em communication range} $r$ is the radius of the
ball
in which a message transmitted
by a station is heard, provided no other station transmits at the same time.
Note that $r=(P/(\cN\beta))^{1/\alpha}$, where
$P$ is the transmission power of a station, c.f., Equation~(\ref{e:sinr}).
Without loss of generality we assume that $r=1$.
(Note that this assumption implies the relationship $P=\cN\beta$.)
%i.e., $(P/(\cN\beta))^{1/\alpha}=1$, where
%$P$ is the transmission power of a station.

\paragraph{Communication graph and graph notation}
The {\em communication graph} $G(V,E)$
of a given network
consists of all network nodes and edges $(v,u)$ such that $\dist(v,u)\leq (1-\eps)r=1-\eps$,
where $0<\eps<1$ is
a fixed model parameter. 
The meaning of the communication graph is as follows: even though the idealistic communication
range is $r$, it may be reached only in a very unrealistic case of single transmission in the whole
network, c.f.,~\cite{DGKN13}. In practice, however, many nodes located in different parts of the network often
transmit simultaneously, and therefore it is reasonable to assume that we may only
hope for a slightly smaller range to be achieved.
The communication graph, through restricting connections to ranges at most $1-\eps$, envisions the network of such ``reasonable reachability''. 
It has become a classic tool in the analysis of ad hoc communication tasks under the SINR physical model
c.f.,~\cite{DGKN13,JKRS13,YuHWTL12}. 
%\dk{WIECEJ CYTOWAN}

Note that the communication graph is symmetric for uniform networks.
By a {\em neighborhood} of a node $u$ we mean the set %(and positions)
of all %in-coming
neighbors of $u$ in $G$, i.e., the set $\{w\,|\, (w,u)\in E(G)\}$.
The {\em graph distance} from $v$ to $w$ is equal to the length of a shortest path from $v$ to $w$
in the communication graph, where the length of a path is equal to the number of its edges.
The diameter $D$ of a network is equal to the diameter of its communication graph
(i.e., the largest graph distance between any pair of nodes),
provided the graph is connected.

\paragraph{Metric space}

Given a metric space with a distance function $\dist$, $B(v,r)$ for a point $v$ from the space and $r>0$
is equal to $\{w\,|\, \dist(v,w)\leq r\}$ and is called a {\em ball} with radius $r$ and center $v$. A {\em unit ball} is a ball with radius $1$. Moreover, let
$\chi(a,b)$ denote the number of balls with radius $b$ sufficient to cover a ball with
radius $a$.
Nodes of a network are embedded (as points) in a general metric space with a distance function $\dist$ that satisfies the following %property. We assume that the metric satisfies the following 
{\em bounded growth property}:
%which corresponds to the following property.
For every $d>0$, $c\in\NAT$ and a point $v$ in the metric space, the ball $B(v,c\cdot d)$ is
included in a union of  $O(c^\gamma)$ balls with radius $d$, where $\gamma$ is
a parameter called a
%the
{\em dimension}
of the metric. (That is, $\chi(cd,d)=O(c^\gamma)$ for each $d>0$ and $c\in\NAT$.)
% and a constant hidden in the big-O notation is its particular feature.
Note that this in particular implies that $B(v,(c+1)\cdot d)\setminus B(v,c\cdot d)$ can be covered
by $O(c^{\gamma-1})$ balls with radius $d$;
we will often rely on this property in our analysis when estimating the total strength of the
interference received at a node.
%(we use this fact in our calculations).
%\tj{
%UWAGA MICHALA: MOZE TE WLASNOSC SFORMULOWAC JAKO PROPOSITION?
%}
%\dk{MYSLE ZE MOZNA ZOSTAWIC TAK JAK JEST - DOSC MOCNO UWYPUKLONE}

\paragraph{Broadcast problem}

In the broadcast problem, %studied in this work,
there is one distinguished node, called the {\em source},
which initially holds a piece of information (also called a {\em source message} or a
{\em broadcast message}).
The goal is to disseminate this message to all other nodes in a network
with connected communication graph.
We are interested in minimizing the {\em time complexity} of this task
being
the minimum number of rounds
after which, for all communication networks defined by some set of parameters,
the broadcast occurs with high probability. %\footnote{high probability} %at least $1-\delta$ for a given $0<\delta<1$.
This time is counted since the source is activated.
%For the sake of complexity formulas, we consider the following parameters:
%$n$,
%$D$ and $\delta$.
%\cc{Mozna napisac w skrocie definicje pozostalych problemow. tj:TODO}

%\dk{We also consider other problems in ad hoc setting, such as wake-up, consensus and leader
%election, c.f., Appendix~\ref{s:applications} for details.}
%%\mr{In the full version of the paper we also consider other problems in ad hoc setting, %%such as wake-up, consensus or leader election.}

\subsection{Previous and related work}
The algorithmic research on communication in the SINR networks started around 10 years ago.
Most papers concentrate on one-hop communication, which includes the local
broadcast problem \cite{GoussevskaiaMW08,HM12,YuWHL11},  link scheduling \cite{K12,HM11},
connectivity \cite{AvinLPP09,HM12SODA} and others.
Among them, the most related to this work are papers on local broadcast, in which each node has
to transmit a message only to its neighbors in the corresponding communication graph.
%If applied $D$
%
Using the local broadcast algorithm (e.g.,\ from \cite{HM12}) as a building block yields a solution for (global) broadcast
that runs in $O(D(\Delta+\log n)\log n)$ time, where $\Delta$ is the maximal degree of the communication graph.
However, since there is only one message to be propagated in the global broadcast,
we would like to avoid the dependence on potentially large parameter $\Delta$
(which could be necessary when all senders have different messages, but not in the case
of a single source global broadcast).

In order to address obstacles for multi-hop communication, various authors
take advantage of several features
helping to design efficient algorithms.
As for the broadcast problem
in the SINR model, Scheideler at al.\ \cite{SRS08} solve the problem in $O(D+\log^2n)$ rounds
using a tunable collision
detection and assuming that all stations start a protocol simultaneously,
%e.g.\ \cite{SRS08},
which allow them to build an overlay structure
along which the message is then propagated.
Yu et al.\ \cite{YuHWTL12} solve the problem in $O(D+\log^2 n)$ rounds using power control,
allowing stations to decide
the strength of a transmitted signal in each step.
Moreover, their results works merely for a restricted family of networks,
excluding the most challenging scenarios.
Specifically, their algorithm works under assumption that, for each node $v$, its closest neighbor is in distance at most $1/3$. Moreover, a possibility of filtering out
messages received from large distances is necessary.

In \cite{JKRS13} an $O(D\log n+\log^2n)$ randomized algorithm and  in \cite{JKS13} an $O(D\log^2n)$
deterministic algorithm for networks deployed in the Euclidean space are presented, where stations know
their own positions (e.g., thanks to GPS devices).
%
%JKRS assume that stations know their positions in the Euclidean space, while
Finally, Daum et al.\ \cite{DGKN13} designed an algorithm working in $O((D\log n)\log^{\alpha+1}R_s)$
rounds, provided stations know only granularity $R_s$ of the network (i.e., the maximum ratio
between actual distances of stations
%which are
connected by an edge in the communication graph) and do not use any other additional features.

If the model prevents successful direct transmissions between nodes which are not
connected in the communication graph, so called {\em weak device model},
the lower bound $\Omega(D\Delta)$ holds even
if stations know their positions on the plane \cite{JKS13i}, which separates that model from
the one considered in this work. For other related problems in this more harsh model see e.g.,~\cite{GoussevskaiaMW08, JK-DISC-12}.
%\tj{WIECEJ CYTOWAN!!!!!!!!!!}

In a related radio network model, the complexity of broadcasting is much better understood.
Its complexity in the model without collision detection is $\Theta((D+\log n)\log(n/D))$ \cite{CzumajRytter-FOCS-03,KP-DC-05}.
Interestingly, this lower bound was recently broken for the model with collision
detection \cite{GHK13}, in which a solution in $O(D+\text{polylog}(n))$ was
designed.
For the
easier case where all nodes start during the same round, it is currently unknown whether or not
formulas better than the ones in general graphs could be obtained,
%better bounds are possible in general graphs,
but in unit disk graphs a solution of the form $O(D + \log^2 n)$ is likely possible
\cite{DGKN13}.
As shown e.g.\ in \cite{GKLW08,FAM13}, geometric graphs exhibit more efficient solutions than those
possible in a general graph model of radio networks.
%
%\tj{mozna cos jeszcze o udg RN, ale te deterministyczne wyniki tutaj mniej pasuja}
%\dk{mozna dodac moja prace z disc 2008 z Lingasem, potencjalnym recenzentem;
%\dk{mozna tez dodac ze w pracy z icalp 2013 pokazujemy ze jesli devices sa weak, tzn. maja
%dodatkowo ograniczona sensitivity, to potrzebny jest czas $\Omega(D\Delta)$ nawet
%dla randomizowanych - tutaj jest jednak maly haczyk ze tam communication graph
%jest dla pelnego range $r$, nie dla frakcji $r$ jak tutaj; ale mimo wszystko mozna krotko
%zagaic roznice}
%Since they are quite efficient, there are very few studies of the problem restricted to
%the geometric setting. However, when mobility of stations is assumed, location and movement
%of stations on the plane is natural. Such settings were studied e.g.,\ in
%\cite{Farach-ColtonM07}.
%\cite{Farach-ColtonAMMZ11,Farach-ColtonM07}.

\subsection{Our results}
The results of this paper state that the broadcast problem can be accomplished in $O(D\log^2 n)$ rounds in the model without spontaneous wake-up and in $O(D\log n +\log^2 n)$ 
rounds in the model with spontaneous wake-up.
Interestingly, this performance formulas does not depend on any geometric parameter
related with specific locations of nodes, only on parameters of communication graph
(expressing the relation whether nodes are within their transmission ranges or not).
This improves the result
of Daum et al.\ \cite{DGKN13} for $R_s=\omega(2^{(\log n)^{1/(1+\alpha)}})$,
where $R_s$ is the maximum ratio between
between actual distances in the metric space of stations which are connected by an edge in the communication
graph.
(Thus, in particular,
for $R_s=\Omega(n^{\delta})$, for any fixed $\delta>0$.) Moreover, our algorithm does not need information
about parameter $R_s$.
%the $R_s$ parameter.
%
As $R_s$ might be even exponential wrt to $n$,\footnote{%
Consider, for example, $n$ stations $x_1,\ldots,x_n$ on a line such that $\dist(x_i,x_{i+1})=1/2^i$.}
this is the first solution with guaranteed $O(D\ \!\text{polylog}(n))$
complexity in the SINR networks without spontaneous wakeup, power assignment, carrier sensing
(tuned collision detection), or
%knowledge of a station of its position in the Euclidean space.
%\dk{
%Our algorithm does not use
any knowledge about location of nodes. % (and therefore
%it also subsumes previous results that required such knowledge, c.f.,~\cite{JKRS13}).}
%}
%
%\dk{Similar formulas can be obtained for other problems in ad hoc setting, such as wake-up, %consensus and leader
%election, c.f., Appendix~\ref{s:applications}.}

%And, up to our knowledge, this is the first efficient
%solution for a multi-hop communication problem in such a weak model. ???

As the main tool, we design a specific coloring algorithm associating with each
active station $v$ the probability $p_v$ which, when fixed, helps to solve other
communication problems efficiently, including the  consensus problem,
the leader election and the alert protocol problem. % (due to limited space, these applications will be presented in the full version of the paper).
This coloring plays a role similar to backbone structures in many other communication models
(c.f.,~\cite{JK-DISC-12,YuWWY13}). 
%\dk{JAKIES CYTOWANIE DO SURVEY O CONNECTED DOMINATING SET
%W ZWYKLYCH SIECIACH LUB UDG?}
%
In order to get rid of the dependence on the granularity parameter $R_s$, we use
a different approach to \cite{DGKN13}. They allow all stations to transmit
with constant probability, which may generate a lot of noise but makes
possible communication between stations within the smallest distance in the network.
We, on the other hand, start from very low probabilities of the order $1/n$
and increase them gradually until stations can hear reasonable number of messages.
Such a strategy reminds solutions to the local broadcast problem (e.g., in \cite{HM12});
however, unlike in those solutions (which did not need it),
the key and subtle issue in our approach
%which does not appear in the local broadcasting
is to somehow distinguish by a station $v$ between the densities of the network in
close neighborhood $B(v,\eps/2)$ and
in broader neighborhood $B(v,1)$.
In order to tackle this issue without any geolocation information and other tools such as power control,
we proceed by interleaving two kinds of phases serving different purposes:
phases where stations transmit with some assigned probabilities (which intuitively grow up gradually)
with phases where probabilities of transmissions are ``scaled up'' based on local statistics
of successful transmissions with carefully probed transmission probabilities.
This approach faces various technical obstacles, mainly due to the lack of geolocation information,
which are addressed in the paper.

\paragraph{Organization of the paper}

Basic properties of simple transmission scenarios are given in Section~\ref{s:preliminaries}.
The main coloring tool, its details and construction can be found in Section~\ref{s:coloring}. 
Its applications to broadcasting in non-spontaneous and spontaneous settings are presented in Section~\ref{s:broadcast}.
Missing proofs can be found in the full version of the paper.
%More details on construction and properties of the coloring are provided in Section~\ref{s:coloring}.
%Missing proofs and technical tools are in Appendix~\ref{s:tools}, \ref{sec:lemma1},~\ref{sec:lemma2}
%and~\ref{s:prop-density}, while applications to other distributed problems
%are in Appendix~\ref{s:applications}.}

%We start 

%solving so-called {\em alert protocol
%problem} \cite{KlonAlbo?}. Given $n$

% \subsection{Technical preliminaries}

\section{Notations and Technical Preliminaries}
\label{s:preliminaries}

We say that an event happens in a network of $n$ stations {\em with high probability
(whp)} when the probability is at least $1-1/n^c$, for some constant $c>0$.\footnote{%
We often use union bounds to show
that some undesirable events happen with small probability.
%do not happen whp,
Therefore, in order to carry on we usually require
for some basic events occurring during the analysis
%analyzed algorithms that they
to happen with probability at least $1-1/n^4$,
when saying that they occur whp. 
%\tj{This requirement will guarantee by the union bound
%that the probability that our algorithms correctly accomplish their tasks in given
%time bounds with probability $\geq 1-1/n$.}
This requirement will guarantee that after applying all union 
bound arguments within the analysis, the probability that our 
algorithms accomplish their tasks correctly within the given 
time bounds is at least $1-1/n$.
}
%\tj{For 
%
%\cc{moze: Where $c$ can be made arbitrarily large, in order to allow for multiple applications of the %union bound.
%This constant is carefully taken into account during the analysis,
%(TJ: NIE ROZUMIEM TEGO ZDANIA; WLASCIWIE NIE JESTESMY OSTROZNI;
%PRZY INDYWIDUALNYCH ZDARZENIACH CHCEMY DUZEGO WYKLADNIKA, TAK ZEBY
%SUMA PBB ZDARZEN NIEPOMYSLNYCH BYLA $O(1/n)$)
%with ``whp'' being used in final statements or places where the constant is clear from the context.
%}
An event occurs with {\em negligible} probability if its negation occurs whp.
In particular we prove in this paper that our algorithms succeed whp.

Given a metric space with the distance $\dist(\cdot,\cdot)$, we use the notation
$B(v,r)=\{v'\,|\, \dist(v,v')\leq r\}$. For a set of stations $A$, by $N(A)$ we denote the
set of their neighbors in the communication graph, i.e.,
$N(A)=\{w\,|\, \dist(v,w)\leq 1-\eps\mbox{ for some }v\in A\}$.

In order to simplify calculations, we assume that the constant hidden in the expressions $O(c^\gamma)$ determining the growth parameter of the metric space is equal to $1$ (this does not change the asymptotic complexity of our
algorithms).

Below, we formulate a basic property that a station transmitting successfully to a distance
larger than $1-\eps$ delivers its message to the neighbors of other stations in its close proximity.

\begin{fact}\labell{f:onbehalf}
If a station $v$ is transmitting in a round and its message can be successfully received
at each point $u$ such that $\dist(u,v)\leq 1-\eps/2$, then the message is received
by all neighbors of all nodes from $B(v,\eps/2)$.
\end{fact}
%Tu moze lemat o tym, ze wystarczy ktos z mojego $\varepsilon/2$ otoczenia transmitujacy
%na odleglosc $1-\varepsilon/2$, zeby powiadomic wszystkich moich sasiadow.

%Oznaczenie $N(v)$ jako $\{u\,|\, d(u,v)\leq 1-\eps\}$.

Given a set of stations $T$ transmitting in a round and a station $u$, the
{\em interference} at $u$ is equal to $I_u=\sum_{w\in T\setminus\{v\}} P/(\dist(u,w))^\alpha$,
where $v$ is a station in the smallest distance from $u$ among the elements of $T$.

\begin{fact}
\label{f:close_int}
Let $x\le1/2^{1/\alpha}$. If the interference at some receiver $u$ is at most $\cN/(2x^\alpha)$, then
it can hear the transmitter $v$ from the distance $x$.
\end{fact}

\begin{proof} Let us recall assumption that $P=\cN\beta$. We have
 \[
  SINR(v,u,T)
  \ge
%  \frac{P}{x^\alpha}/(\cN+\frac{\cN}{2x^\alpha})
%  \ge
  \frac{\beta\cN}{x^\alpha}/(\cN+\frac{\cN}{2x^\alpha}) = \beta/(x^\alpha+1/2) \ge \beta
  \ ,
 \]
where the former inequality follows from the bound on the interference and the latter from
the assumption $x\le1/2^{1/\alpha}$.
\end{proof}

\begin{fact}
\label{fact01}
If the interference at some receiver $u$ is at most $\cN\alpha x$, then
it can hear the transmitter $v$ from the distance $1-x$.
\end{fact}

\begin{proof}
%\noindent
%{\bf\em Proof of Fact~\ref{fact01}:}
By the Bernoulli inequality we get
$(1+x)^\alpha \ge 1+\alpha x$. Thus
\[
SINR(v,u,T)
\ge
\frac{P/(1-x)^\alpha}{\cN+\cN\alpha x}
\ge
\frac{P}{\cN (1+x)^\alpha(1-x)^\alpha}
=
\frac{P}{\cN (1-x^2)^\alpha}
\ge
\frac{P}{\cN}=\beta
\ .
\]
%where the last equality follows from the assumption that
%$r=(P/(\beta\cN))^{1/\alpha}=1$.
%the range of stations is equal to $1$ which implies $\left(\frac{P}{\cN\beta}\right)^{1/\alpha}=1$.
\end{proof}
% maksymalna interferencja do słyszalności w odl. x

%\tj{wyrzucamy ten lemat; chyba ze ktos zaprotestuje}
%\begin{lemma}
%If $\sum_{v\in B} p_v < C_1$ for every unit disc $B$ then the expected interference induced at $v$
%by stations within distance at least $d$, denoted by $I(v,d)$, can be bounded
% by $I_d$.
% \tj{a gdzie $I_d$? a co to jest $I(v,d)$?}
%\dk{uzylbym tutaj i w dowodzie $c_1$ zamiast $C_1$, bo lepiej zafiksowac duze "C" na kule;
%albo zmienic notacje na kule z "C" na "B"?}
%\end{lemma}

Consider a scenario where every station $v$ is
assigned a variable $p_v$ being its transmission probability.

\begin{fact}\labell{f:sum}\cite{JS05}
Assume that $\sum_{v\in A}p_v=s\leq 1/2$ for some set of stations $A$.
Then the probability that exactly one element of $A$ transmits
is at least $s/2$ and at most $s$.
\end{fact}

\begin{fact}\labell{f:notransmit}\cite{JS05}
Let $p_v\le1/2$ for every station $v\in A$. Then the probability that no
station from $A$ transmits is at least $(1/4)^{\sum_{v\in A}p_v}$.
\end{fact} 

We say that {\em bounded density property} is satisfied with the parameter $C>0$
if %the following condition is satisfied:
%There exists constant $C>0$ such that
%if
$\sum_{w\in B} p_w \le C$ for every unit ball $B$.
  %then
%\end{fact}
The {\em effective communication property} is satisfied if the probability that a station $v$ hears $w$ when $w$ is the only transmitting station in $B(v,2/3)$ is at least $1/2$.

\begin{fact}\label{fact7}
For any network parameters $\alpha>\gamma$, $\beta\ge 1$, $\cN>0$ and $\eps<1$,
there exists a constant $C>0$ such that  if  the bounded density property
is satisfied with any parameter $C_1<C$, then
the effective communication property is satisfied as well.
%$$Pr(\text{$v$ hears $w$ $|$ $w$ is the only transmitting station in }B(v,2/3)) \ge 1/2.$$ %c_D$.
\end{fact}
\begin{proof}
Assume that $w$ is the only transmitter in $B(v,2/3)$
and
$\sum_{w\in B} p_w \le C$ for every unit ball $B$
and a constant $C$.
Then, $v$ receives the message from $w$  by Fact~\ref{fact01},
provided the interference from the remaining area is smaller than $\cN\alpha\cdot \frac13$.
The expected value of this interference under the bounded density property with the parameter $C$ is
$$\begin{array}{rcl}
E(I_v) & \leq  & C\sum_{i>0}\frac{P\cdot O(i^{\gamma-1})}{\cN\cdot(\frac23i)^\alpha}=\frac{CP(3/2)^\alpha}{\cN}O(\sum_{i>0}i^{\gamma-\alpha-1})
\leq  C\cdot C'
\end{array}$$
where $C'$ is a constant depending on $\gamma$, $\alpha$, $\beta$, $\cN$. 
The first among the above inequalities follows from the bounded growth 
property and the last one from the assumption that $\gamma<\alpha$.
Thus if $C<\frac{\frac16\cN\alpha}{C'}$ then $E(I_v)\leq
C\cdot C'<\frac12\cdot \frac13\alpha$.
Using Markov bound, we get $P(I_v>\cN\alpha\cdot \frac13)<1/2$ and the probability that
$v$ receives a message from $w$ is at least $1/2$ by Fact~\ref{fact01}.
\end{proof}

\section{Network Coloring}
\label{s:coloring}

The key ingredient of broadcasting algorithms presented in this paper is the procedure
StabilizeProbability
(Algorithm~1),
%presented in detail later in Section~\ref{s:coloring}),} 
which assigns probability (``color'') $p_v$
from the set $$\{2^i p_{\text{start}}\,|\, i\in[0,\lceil\log(p_{\max}/p_{\text{start}})\rceil \}$$ to each station $v$ participating in an execution, where $p_{\text{start}}=\Theta(1/n)$ and $p_{\max}$ is a constant which will
be specified later. Thus, the number of colors is $O(\log n)$.

Before giving details of the procedure StabilizeProbability, we state the key properties that we want the procedure to satisfy.
%The procedure StabilizeProbability works in $O(\log^2 n)$ rounds, \dk{c.f., Fact~\ref{f:SP-time}.}
%it is described and analyzed in the next section of the paper.
%
We express them as properties of the obtained coloring, described by the following lemmas. These lemmas are true for some constants $C_1$ and $C_2$
%and %$c_\eps$ 
that depend on $\eps$, $\gamma$, the parameters of the SINR model and constants chosen in the algorithm.
%\cc{Skoro wszystkie stale $C_1,C_2,c_\eps$ zaleza od $\eps$, $\gamma$ and the parameters of the SINR %model
%to po co nam az trzy a nie dwie?}

\begin{lemma}
\labell{lemma1}
After an execution of StabilizeProbability for a set of stations $A$, the inequality
$$\sum_{\substack{w: p_w=p \\ w \in B}} p_w < C_1$$
holds for every color $p$ and unit ball $B$ whp.
\end{lemma}

\begin{lemma}
\labell{lemma2}
After an execution of StabilizeProbability on a set of stations $A$, for every
$v\in A$ there exists a color $p$ such that the following inequality holds whp:
$$\sum_{\substack{w: p_w=p \\ w \in B(v,\eps/2)}} p_w \ge C_2 \ . $$
%$$\sum_{\substack{w: p_w=p \\ w \in \dk{B(v,\eps/2)}}} p_w \ge C_2/c_\eps \ . $$
\end{lemma}

In this section we describe formally the algorithm StabilizeProbability and prove
that it satisfies Lemma~\ref{lemma1} and Lemma~\ref{lemma2} for appropriate
constants $C_1$ and $C_2$. %,C_2$ and $c_{\eps}$.
Its pseudo-code is given as Algorithm~1. %\ref{alg:stab}
The pseudocode
is missing information about the actual values of constants $c_0$, $c_1$, $c_2$, $c_3$,
$c'$, $c_\eps$, $C_2$, and $p_{\text{start}}$ used in the algorithm.
We will choose the appropriate values for those constant in the analysis of
the algorithm and its properties.

Algorithm StabilizeProbability performs two kinds of tests, defined by sub-routines
DensityTest and Playoff, each of them taking $O(\log n)$ rounds.
%These procedures are executed $O(\log n)$ times, until both are successful;
%each time the input probability is increased by factor $2$, starting from $\Theta(1/n)$
%and ending by at most constant probability $p_{max}$. 
The while loop of StabilizeProbability is repeated $O(\log n)$ times, since
each node $v$ starts with $p_v=\Theta(1/n)$, increases $p_v$ twice in each repetition
of the loop and finishes either in line 6 or after achieving $p_v\geq p_{\max}=\Theta(1)$.
Therefore, the following claim holds.

\begin{fact}
\label{f:SP-time}
Algorithm StabilizeProbability works in $O(\log^2 n)$ rounds whp.
\end{fact}

In Section~\ref{sec:col:over},
we give a general idea why
%a general idea explaining how we are going to assure that
StabilizeProbability satisfies Lemma~\ref{lemma1} and Lemma~\ref{lemma2}.

In Section~\ref{sec:density}, the constants $c_0, c_1$ in the
sub-routine
%procedure
DensityTest are set so that this procedure
helps each station $v$ to estimate whether density (i.e., the sum of probabilities assigned to stations)
in $B(v,1)$  already achieved
constant value. %As a result, we also fix $C_1$ from Lemma~\ref{lemma1}.
%\cc{Nie rozumiem jak ostatnie zdanie ma sie do nastepnego paragrafu o Lemma 1.}

%\mr{Due to the technical complexity of 
Proofs of Lemma~\ref{lemma1} and Lemma~\ref{lemma2} %we state them 
are presented in separate sections. %}
In Section~\ref{sec:lemma1}, containing the proof of Lemma~\ref{lemma1},
the constants $c_2$, $c_3$ for sub-routine Playoff
are chosen in order to guarantee that the regions (unit balls) with largest density will
become sparser, i.e., some stations switch off (line 6 of Algorithm~1) in them in each execution of Playoff whp.
Combined with the appropriate choice of the constant $c'$ from Playoff, this will assure
that Lemma~\ref{lemma1} is satisfied.

In Section~\ref{sec:lemma2},
dedicated to the proof of Lemma~\ref{lemma2},
the constant $c_\eps$ is chosen in order to make very unlikely the situations in which a station $v$ switches off when the sum of probabilities of active stations in 
%the $(\eps/2)$-neighborhood of $v$ 
$B(v,\eps/2)$ is very small. This property, combined with the fact that probabilities of 
``active'' stations (i.e., stations which are not switched off in line 6 of Algorithm~1) grow up to the constant $p_{\max}$, will lead to the statement of Lemma~\ref{lemma2}.

\subsection{Overview of the algorithm}\labell{sec:col:over}
%\tj{to gdzie indziej}
%Let $A$ be the set of active stations.

%\tj{ten akapit byc moze powiela opis Grzeska; zweryfikowac}
%In the following, we assume that given a set of stations $S$ with assigned probabilities
%$p_v$ such that $\sum_{v\in S}p_v\leq 1/2$, the probability that exactly one element of $S$
%transmits when all of them choose transmitting with assigned probabilities is of the order of
%$\sum_{v\in S}p_v$.
%\mr{Observe, that by Fact~\ref{f:sum} given a set of stations $S$ with assigned probabilities $p_v$ such that 
%$s=\sum_{v\in S}p_v\leq 1/2$ the probability that exactly one of them transmits, when each station $w$ transmits
%with pbb. $p_w$, is between $s/2$ and $s$.}

%The core of all results in this paper is 
The coloring algorithm assigns probability/color $p_v$ to each
active station $v\in A$ such that there are at most $\log n$ various colors and simultaneously
the following two properties hold:
\begin{enumerate}
%\item

\item
For each color, the sum of probabilities of stations in this color in each unit ball is at most $C_1$, for some constant $C_1$
(Lemma~\ref{lemma1}).
\item
For each active station $v$, there exists a color such that the sum of probabilities of this color in
the ball $B(v,\eps/2)$ is
at least $C_2$, for some constant $C_2$
(Lemma~\ref{lemma2}).
\end{enumerate}
The former property assures that, when all stations transmit with assigned probabilities,
the expected interference coming from the whole network
is small at any station
(this follows from the assumption that $\alpha>\gamma$, which implies $\sum_{i\in\NAT}i^{\gamma-1-\alpha}=O(1)$).
Thanks to this property, if a station $v$ is the only transmitter in $B(v,2d)$ for %(not too large)
a constant
$d\leq 1-\eps$ then the message transmitted by $v$ can be received in each point of $B(v,d)$
with constant probability.
The latter property on the other hand guarantees that for each station $v$, the probability
that a station from $B(v,\eps/2)$ transmits is constant as well, for some color.
Both properties combined imply that, each station from $N(A)$ receives a message with
probability $\Omega(1/\log n)$ in a round if each station $v\in A$ transmits with probability $p_v/\log n$.

There is some intuition behind the coloring algorithm. The optimal probability for a station
$v$ to transmit is approximately $1/|\{w:\dist(v,w)\leq \eps/2\}|$. This would ensure that
the sum $\sum_{B(v,\eps/2)}p_v$ is limited from below by constant $C_2$ and $C_1$ is of the order of $C_2/\eps^\gamma$. %\tj{(or $1/\eps^\gamma$ where $\gamma$ is the growth-bound of our metric).}
%\cc{W poprzednim zdaniu: dlaczego "some constant $C_2$", a nie po prostu "constant $C_2$"?}
%
The algorithm starts from low probabilities (smaller than $C_1/n$), continuously increasing them.
Once a station $v$ starts receiving messages from others, it assumes that the sum of probabilities
in $B(v,1)$ is constant, which indicates that further increase of all probabilities in $B(v,1)$
could break the property 1 (Lemma~\ref{lemma1}).
When applied to the task of local broadcasting \cite{HM12}, this means that the ``right''
probability is reached (so, probabilities could be frozen).
%
%\cc{jak taka strategia zadziala na wykladniczej linii.}
%\tj{o tym, ze suma pbb dobrze aproksymuje pbb uslyszenia wiadomosci}
%
If stations are uniformly distributed, then this implies that the average sum of probabilities
in balls of diameter $\eps$ is around $C_1/\eps^{\gamma}$, which would also satisfy the property
2 (Lemma~\ref{lemma2}) for $C_2\approx C_1/\eps^\gamma$.
In general, the constant sum of probabilities, around $C_1$, of stations in $B(v,1)$ does not exclude
that the sum of probabilities in $B(v,\eps/2)$ is still very small.\footnote{%
Consider for example stations $v_1,\ldots,v_n$ on a line,
where the distance between $v_i$ and $v_{i+1}$ is $1/2^i$, for $1\le i<n$.}
Hence, our intuitive goal at this stage would be to distinguish those regions (balls of diameter
$\eps$) in which the sum
of probabilities exceeds (say) half of the average from those where it is much smaller.
%\cc{nie rozumiem koncowki}

The main difficulty is to sense the actual sum of probabilities in $B(v,\eps/2)$,
without possibility of filtering out messages received from distance larger than $\eps/2$
(as there is no geolocation).
Imagine %for a while
that stations in a network are uniformly distributed and the probabilities
have reached such values that the sums in a unit ball are close to $C_1$.
%Then,
If one replaces $p_v$ with $c_\eps p_v$ for large enough $c_\eps$ (depending on the
growth parameter of the metric, e.g., $c_\eps\approx 1/\eps^2$ in the Euclidean plane),
then the average sum of probabilities in an
$\eps$-ball
%$\eps$-disc
is around $C_1$.
%
%This simplified observation leads to the idea that stations
Hence, an ``average'' station still receives the number of messages similar to those received
with
the original probabilities $p_v$.
%non-increased probabilities.
%
If, however, the stations are not uniformly distributed, it is still the case that stations in the
smallest ball with sum of probabilities at least $C_1$ could receive many message
%when
after probabilities $p_v$ are scaled up to $c_\eps p_v$.
%pbb are increased (scaled up?).
On the other hand, the situation in a (very) sparse
$\eps$-ball $B$
%$\eps$-disk $C$
is as follows
(by ``sparse
$\eps$-ball'' we mean a ball
%$\eps$-disk we mean a disc
with sum of probabilities much smaller than the average):
\begin{itemize}
\item
The probability that $v\in B$ receives a message from other station from $B$ is small
(the sum is so small that usually no one is transmitting);

\item
The probability that $v\in B$ receives a message from $u\not\in B$ is small as well
(the sum in $B(v,1)$ is as large after scaling up by constant $c_\eps$ that usually the 
%noise
interference
%\cc{ noise or interference? }
prevents any successful transmission from distance larger than $\eps/2$).
\end{itemize}
Using this idea, our coloring algorithm works as follows.
Procedure DensityTest verifies whether the sum of $p_v$'s in a unit-disk around
a station $v$ is close to $C_1$
(i.e., whether $v$ receives many messages
when transmissions occur with probabilities $p_v$).
%with current probabilities).
%
If it is the case, procedure Playoff verifies if the density in close proximity of $v$
is large (i.e., whether $v$ still receives many messages when probabilities
are scaled up).
In the case of positive outcomes of both procedures, $v$ is switched off, which decreases the sum
of probabilities in the unit disk around $v$.
If repeated sufficient number of times, Playoffs allow to preserve property 1 (Lemma~\ref{lemma1}),
i.e., prevents the sums of probabilities in unit disks from going above $C_1$.
On the other hand, as the positive result of Playoff cannot happen in (very)
sparse areas while probabilities of stations (if not switched off) grow up
to the constant $p_{\max}$, the property 2  (Lemma~\ref{lemma2}) is preserved %as well.
at the end of an execution of StabilizeProbability.

%\subsection{Algorithm}

%TODO Ujednolicic wyglad procedur i algorytmu

\begin{algorithm}[H]
  \begin{algorithmic}[1]
    \Procedure{DensityTest}{$v$}
      \For{$c_0\log n$ rounds} transmit with prob. $p_v$ \EndFor
      \If{received at least $c_1\log n$ messages} return True
      \Else{ return False}
      \EndIf
    \EndProcedure
  \end{algorithmic}
\end{algorithm}

\begin{algorithm}[H]
  \begin{algorithmic}[1]
    \Procedure{Playoff}{$v$}
      \For{$c_2\log n$ rounds} transmit with prob. $p_v\cdot c_{\varepsilon}$ \EndFor
      \If{received at least $c_3\log n$ messages} return True
      \Else{ return False}
      \EndIf
    \EndProcedure
  \end{algorithmic}
\end{algorithm}

\begin{algorithm}[H]
  \begin{algorithmic}[1]\label{alg:stab}
  \caption{StabilizeProbability($v$)}
    \State $p_v \leftarrow p_{start}$\Comment{$p_{start}=C_1/(2n)$}
    \State $p_{max} \leftarrow C_2/c_\eps$
    \While{$p_v < p_{max}$}
    \For{$c'$ times}
      \If{DensityTest($v$) and Playoff($v$)}
	\State $v$ quits with color $p_v$
      \EndIf
    \EndFor
    \State $p_v \leftarrow 2p_v$
    \EndWhile
    \State $v$ quits with color $2p_{max}$
  \end{algorithmic}
\end{algorithm}

\subsection{DensityTest}\labell{sec:density}
In this section we fix the constants $C_1$, $c_0$ and $c_1$ and state
properties of DensityTest which are satisfied for this choice of
constants.

%We say that assignment of probabilities $p_v$ to stations satisfies
%{\em general density assumption} if the following condition is satisfied:
%if a station $v$ is

%In this section we analyze behavior of the $DensityTest$ routine.

%\begin{fact}

%\begin{equation}\label{ec}
%Pr(\text{$v$ hears $w$ $|$ $w$ is the only transmitting station in }B(v,2/3)) \ge 1/2.
%\end{equation}

From now on assume that $C_1$ is any value such that the bounded density property holds with such $C_1$ that the effective communication property is satisfied as well.
The goal is to choose $c_0$ and $c_1$ such that the probability of receiving a successful
transmission is %(relatively) high 
around $c_1/c_0$
in a unit ball with the sum close to $C_1/2$ and
it is much lower than $c_1/c_0$ if the sum is significantly smaller than $C_1/2$.
As possibility of sensing stations in distance close to $1$ heavily depends on network
topology and because of some technical reasons, the actual properties (provably) guaranteed
by our choice of $c_0$ and $c_1$ will be a bit different from this intuitive goal.

\begin{proposition}
\labell{prop:density}
Assume that bounded density property is satisfied with
the parameter $C_1$ guarantying the effective communication property.
Then, one can choose $c_0, c_1$ and $c_d$ such that, for every node $v$, 
the following properties are satisfied:
\begin{enumerate}
\item[(1)]
Let $y=\chi(1/6,1)$ be the number of balls of radius $1/6$ sufficient to cover a unit ball.
If
$1/2\geq \sum_{w\in B(v,2/3)}p_w\geq C_1/(2y)$
then the routine DensityTest$(v)$ returns True whp.
%with high probability.
%\end{proposition}
%
%\begin{proposition}
%\label{prop:densfalse}
\item[(2)]
There exists a constant $c_d$, such that if
$\sum_{w\in B(v,1)}p_w< C_1 c_d$
then the routine DensityTest$(v)$ returns False with high probability.
\end{enumerate}
\end{proposition}
The above proposition shows that DensityTest gives an opportunity to distinguish
areas with large sums of probabilities from those with 
much
%more 
smaller sums of probabilities.
%The proof can be found in full version of this paper. %Appendix~\ref{s:prop-density}.

\def\ProofPropDensity{%%%%%%%%%%%%%%%%%%%
%\begin{proof}
Let success in a round of DensityTest($v$) means that $v$ successfully receives or
sends a message in that round.

\noindent(1)~Let $C'_1=C_1/(2y)$.
By the effective communication property, the probability of success is at least
$\frac12\cdot \sum_{w\in B(v,2/3)}p_w/2\geq C'_1/4$ (cf.\ Fact~\ref{f:sum}). If one chooses
large enough $c_0$ and $c_1$ such that $c_1/c_0\leq \frac12\cdot C'_1/4$, then
the probability of success in a round is at least $2\cdot\frac{c_1}{c_0}$.
And, using a standard Chernoff bound, the result of DensityTest is True whp.

\noindent(2)~For $c_0$ and $c_1$ chosen before, we adjust the constant $c_d$
such that the second claim is satisfied. As a station cannot receive a message
from distance larger than $1$, the probability of success %that $v$ receives a message
is not larger than the probability that at least one station from $B(v,1)$
is transmitting in a round. This probability is bounded from above by
$\sum_{w\in B(v,1)}p_w$. Therefore, if
\begin{equation}\label{e1}
\sum_{w\in B(v,1)}p_w<\frac12\cdot\frac{c_1}{c_0}
\end{equation}
then the result of DensityTest is False whp, by a Chernoff bound.
As %$\sum_{w\in B(v,1)}p_w<C_1$ by the general density assumption
%and
$c_1/c_0\approx \frac12\cdot C'_1/4=C_1/(16y)$, the
claim holds for $c_d<1/(16y)$.
%\end{proof}
} %%%%%%  END \def\ProofPropDensity{

\ProofPropDensity

%\tj{(1) w powyzszym proposition zostalo dopasowane w taki sposob, zeby
%potem wykorzystac w dowodzie ``Lematu 1'' w nastepujacy sposob:
%skoro suma w $B(v,1)$ jest wieksza niz $C_1/2$, to musi byc w nim kolo o
%promieniu $1/6$ z suma $C_1/(2y)$. W nim szukamy najgestszej kratki,
%a potem kratka ucieka ale pozostaje w kole o promieniu $2/3$
%zawierajacym to geste kolo o promieniu $1/6$; wiec DensityTest
%w tej nowej kratce tez zwroci true.}

\def\lemmaone{
\subsection{Proof of Lemma~\ref{lemma1}}
\labell{sec:lemma1}
%\subsection{Proof of Lemma~\ref{lemma1}}\labell{sec:lemma1}

\comment{%%tego nie potrzebujemy
To guarantee that the sum of probability never exceed $C_1$ in every unit ball, we analyze the upper bound on every ball
of radius $\frac{2}{3}$. Since by the following Lemma we bound the probability mass in every ball with radius $\frac{2}{3}$ by
$C_1' = C_1/\chi(\frac{2}{3},1)$ we get the bound $C_1$ on a unit ball by covering it with $\chi(\frac{2}{3},1)$ smaller balls.
}

In this section we prove Lemma~\ref{lemma1}, assuming that $C_1,$ $c_0$ and $c_1$ are the constants
satisfying properties stated in Proposition~\ref{prop:density}.
Moreover, we determine values of $c_2$ and $c_3$, which depend on $c_\eps$ (and $c_\eps$ can be arbitrary at this stage).

Recall that an execution of lines 4--7 of StabilizeProbability is called a {\em phase}.
As the initial probabilities are set to $p_{\text{start}}$, we have $\sum_{w\in B(v,1)}p_w\leq C_1/2$ for each $v$ at the beginning of the algorithm.
Therefore, it suffices to show that 
$\sum_{w} p_w\leq C_1$ at the end of a phase,
provided the same inequality is satisfied at the beginning of this phase, where the sum is taken over $w$ that are active in a given round.
As the probabilities of active stations are multiplied by $2$ at the
end of each phase, the above condition for correctness of Lemma~\ref{lemma1} 
can be deduced from the following lemma.

%
%\mr{mowiac o mass of the probability, sum of probabilities, w jakims regionie mamy na mysli %$\sum_{w\in A}p_w$ a nie $\sum_{w\in A}p_w c_\eps$}

\begin{lemma}
\labell{l:turningoff}
For every $v$, if before ``For $c'$ times'' loop the following inequality holds
$$\sum_{\substack{w \text{ is active} \\ w\in B(v,1)}} p_w < C_1$$ then after the loop
$$\sum_{\substack{w \text{ is active} \\ w\in B(v,1)}} p_w < \frac{C_1}{2}.$$
\end{lemma}
The remaining part of this section is devoted to the proof of Lemma~\ref{l:turningoff},
which in turn follows from the following property, provided $c'$ is chosen large enough.

\begin{lemma}
\labell{l:playoff}
%\mr{przenioslem to zdanie z poprzedniego lematu, bo chyba mialo byc tutaj:}
%\tj{
There exists a constant $q$ which satisfies the following statement.
%}
For every $v$ if $\sum_{w\in B(v,1)} p_w \ge C_1/2$ then there is a set of stations $S \subseteq B(v,4/3)$ such that
$\sum_{w\in S} p_w c_\eps \ge q$ and every station from $S$ is turned off
in line 6
 of the algorithm whp.
\end{lemma}

Now, we prove Lemma~\ref{l:turningoff} assuming correctness of Lemma~\ref{l:playoff} and then we give the proof
of Lemma~\ref{l:playoff}.
By Lemma~\ref{l:playoff} the sum of probabilities in $B(v,4/3)$
will decrease by at least $q/c_\eps$ in every iteration of the ``For $c'$ times'' loop. The maximal sum of probabilities in $B(v,4/3)$ is at most $\chi(1,\frac{4}{3}) C_1$.
Thus by performing $c' =\chi(1,\frac{4}{3}) C_1 c_\eps / q$ iterations we have the  ''opportunity'' to reduce all the
probabilities in $B(v,4/3)$, as long as $\sum_{w\in B(v,1)}p_w\ge C_1/2$.
Thus after $c'$ iterations of the loop we have $\sum_{w\in B(v,1)}p_w < C_1/2$.

%\begin{proof}
 %Now we will prove the Lemma.

It remains to prove Lemma~\ref{l:playoff}; the proof is presented in the following part
of this section.
As we mainly analyze Playoff below, where stations transmit with probabilities scaled up by the
factor $c_\eps$, we use the notion of {\em mass of probability}
of some set of stations $A$ as $c_\eps\sum_{w\in A}p_w$.

%CZY PRZYDATNY TAKI WSTEPNIAK?
The proof of Lemma~\ref{l:playoff} requires to show that, close to each dense unit ball, a group of stations $S$
with probability mass $\geq q$ exists,
for which DensityTest and Playoff return true whp.
%As for DensityTest, we ensure that stations from $S$ satisfy the
%assumption (1) from Proposition~\ref{prop:density}.
%
The main effort in the proof is in ensuring that %such a group $S$
%returns True in Playoff whp, 
the elements of $S$ can hear a message
with constant probability $p$ in each round of Playoff (i.e., when
the probabilities are scaled up by $c_\eps$) and with probability
$\geq c_1/c_0$ in each round of DensityTest (i.e., with ``standard''
probabilities $p_v$'s). 
(When this property
is shown, one can adjust the value $c_3/c_2$ to $p$.)
We show existence of such $S$ by first proving that, in neighborhood of a dense
unit ball, a ball with probability mass $\geq q$, center $x$ and of radius $r$ exists,
which satisfies the following properties for some $b$ (Lemma~\ref{l:playoff2}):
\begin{enumerate}
\item[(a)]
the probability mass of each ball of radius $r$ inside $B(x,br)$ is at
most $z^\gamma q$;
\item[(b)]
$b$ is large enough to guarantee that the number of balls of radius $r$
necessary to cover $B(x,br)$ is such that they can accumulate the whole
probability mass $C_1$ of the unit ball, provided each of them has the
(maximal) mass $z^\gamma q$, where $z>2$ is some constant;
\item[(c)]
$br\leq 1/6$;
\item[(d)]
the sum of probabilities of stations in the ball of radius $2/3$ concentric with $B$
is $\geq C_1/(2y)$, where $y=\chi(1/6,1)$.
\end{enumerate}
Using the above properties, the chances of receiving a message by a station
during Playoff are estimated in Lemma~\ref{l:receiving}. Then, $c_2$ and $c_3$ are chosen
appropriately, to assure that Playoff returns true whp in $S$. On the other hand, (d) above
guarantees that DensityTest returns true whp in $S$.

%\mr{``We show that there exists constants $a, b, q...$ such that the following lemmas hold'' czy definiujemy je tutaj?}

Before stating the following technical lemma sketched by (a)-(d) above, we estimate the value of
$b$ satisfying the condition (b). It is sufficient that $b^{\gamma}z^\gamma q\geq C_1 c_\eps$ which
means that $b\approx \frac1{z}(C_1c_\eps/q)^{1/\gamma}$ is suitable.

\begin{lemma}%[uciekajaca kratka]
\labell{l:playoff2}
 For every $v$, whenever $\sum_{w\in B(v,1)} p_w \ge C_1 / 2$, there exists $x$ such that $B(x,r)\subseteq B(v,4/3)$ and $r\le (\frac{2q}{c_\eps C_1})^{1/\gamma}$ such that if we denote
 \begin{enumerate}
   \item[] $D_0=B(x,r)$
   \item[] $D_1=B(x,ar)\setminus D_0$ for some $a\leq b$
   \item[] $D_2=B(x,br)\setminus (D_1\cup D_0)$
\end{enumerate}
then
 \begin{enumerate}
\item[(1)]
The mass of probability in $D_0$ is at least $q$ and at most $1/2$.
\item[(2)]
For every $x'\in D_1\cup D_2$ the mass of probability in $B(x',r)$ is bounded: $\sum_{w\in B(x',r)} p_w \le z^\gamma q$.
\item[(3)]
For all $w\in D_0$ we have $\sum_{u\in B(w,\frac{2}{3})} p_u \ge C_1/(2 \chi(\frac{1}{6},1))$.
% z^\gamma
\end{enumerate}
\end{lemma}
%proof of lemma uciekajaca kratka
\begin{proof}(of Lemma~\ref{l:playoff2})
Let $v$ be an arbitrary vertex such that
$\sum_{w\in B(v,1)} p_w \ge C_1 / 2$.
Let $B$ be a ball with radius $1/6$ included in $B(v,1)$
with the largest mass of probability.
Thus, $\sum_{w\in B}p_w\geq C_1/(2y)$
(cf.\ Proposition~\ref{prop:density}),
%there is a ball $C\subset B(v,1)$ of radius $1/6$
where $y=\chi(1/6,1)$.
Observe that, if
$x$ located in a ball of radius $2/3$ concentric with $B$
satisfies (1) and (2), then (3) is satisfied for $x$ as well.
Therefore, the idea of our proof is to start looking for
$x$ satisfying (1)--(3) in $B$, as defined above (see Fig.~\ref{fig:kratka}).

\begin{figure}
\begin{center}
\epsfig{file=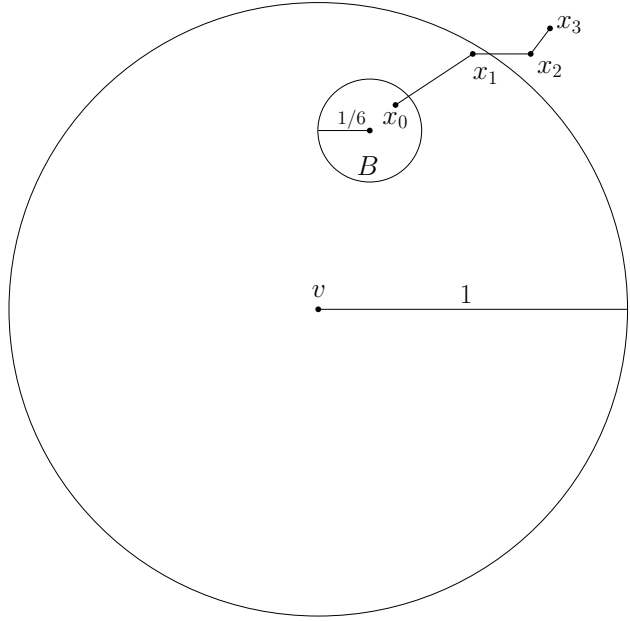, scale=0.6}
\end{center}
\vspace*{-3ex}
\caption{%
An illustration for the proof of Lemma~\ref{l:playoff2}.
The distances $\dist(x_i,x_{i+1})$ form a geometric sequence
$br, br/z, br/z^2,\ldots$ $B$ is the ball
with the largest probability mass among balls of radius $1/6$
included in $B(v,1)$.
}
\label{fig:kratka}
%\vspace*{-2ex}
\end{figure}%

Let $r_0$ be a number satisfying the relationship $C_1/2=q/r_0^\gamma$.
That is the average probability mass of a ball of radius $r_0$ in $B(v,1)$ is
at least $q$.
As $B$ is chosen to have the largest probability mass among balls of
radius $1/6$ included in $B(v,1)$, the average probability mass of a ball of radius $r_0$ in $B$
is at least $q$ as well.
Then, $B$ includes a ball $B_0=B(x_0,r_0)$ with probability mass $\geq q$ such
that $r_0=(2q/(C_1 c_\eps))^{1/\gamma}$. If (2) is satisfied for $x_0$, (3) holds as
well by the choice of $B$.
On the other hand, if (2) is not satisfied for $x_0$, there is a ball of radius $r_0$ and
probability mass $\geq z^\gamma q$ in distance at most $br_0$ from $x_0$.
Bounded growth property of the metric guarantees that this ball contains
$B_1$ of radius $r_1=r_0/z$ and probability mass $\geq q$.

%\tj{(provided $z$ is larger than the base of exponent in the growth bounding
%expression of the metric).}
One can build in such a way a sequence of balls $B_0, B_1,\ldots$ with probability
mass larger than $q$, such that the radius of $B_i=B(x_i,r_i)$ is $r_i=r_{i-1}/z=r_0/z^i$ for $i>0$,
and the distance between (centers of) $B_i$ and $B_{i+1}$ is $\leq br_0(1+1/z+\cdots+1/z^i)$,
as long as (2) is not satisfied.
%when $x=x_i$ is chosen to be the center of $B_i$ and $r=r_i$.
%
If (2) is eventually satisfied for some $x_i$ and $r_i=r_0/z^i$, we obtain a ball $B_i$
of radius at most $r_0=(2q/C_1)^{1/\gamma}$ with probability mass $\geq q$, whose
center point is in distance at most $br_0\sum_{i\geq 0} 1/z^i\leq 2br_0$ from $x$, provided
$z>2$.
We call such an event {\em success}.
Thus, %as \tj{$2br_0\leq 1/3$ (GDZIE USTALAMY $a$ i $b$???)} 
the circle concentric with $B_i$ of radius $2br_0+1/6$ contains
the ball $B$ with sum of probabilities $\geq C_1/(2\chi(\frac{1}{6},1))$, 
On the other hand, $$br_0\leq \frac1{z}(C_1c_\eps/q)^{1/\gamma}\cdot \left(\frac{2q}{c_\eps C_1}\right)^{1/\gamma}\leq 1/z\leq 1/6$$
for $z\geq 6$.
This implies that $2br_0+1/6\leq 2/3$ and therefore (3) is satisfied
as well for $x_i$ and $r_i$.

It remains to show that {\em success} eventually appears in construction of the above
sequence of balls. Note that if $br_i$ is smaller than half of the smallest distance
between stations, there is at most one station in a ball of radius $br_i$. This in turn
implies that, if the probability mass of $B(x,r_i)$ is nonzero, then
the probability mass of $B(x,br_i)\setminus B(x,r_i)$ is zero.
Therefore (2) is satisfied for $x$ and $r=r_i$.
%
%around each station $x$ (excluding $x$) is equ
%
\end{proof}

%\mr{Tutaj podobnie -- czy ustalamy $c_2$ i $c_3$, czy dopiero po wszystkich Claimach napiszemy, ze wybieramy
%takie $c_2$ i $c_3$, zeby pasowalo do p(v), ktore policzylismy?}

\begin{lemma}
\labell{l:receiving}
  Let $D_0$ be a ball satisfying assumptions of Lemma~\ref{l:playoff2} and conditions (2) and (3) stated in this lemma.
  %For every set $D_0$, defined and considered under conditions from Lemma~\ref{l:playoff2}, satisfying \ref{l:playoff2}.1, \ref{l:playoff2}.2,
  Then, for every $v\in D_0$ the probability of receiving a message $p(v)$ is at least %$2c_3/c_2$.
  $q/8\cdot (1/4)^{a^\gamma z^\gamma q}$. %= 2c_3/c_2$.
\end{lemma}

\begin{proof}
 Now we analyze the probability that every station $v\in D_0$ receives a message. Let $D$ denote the area around $D_0$, and $\cI_0$ be the interference allowing for transmission on a distance $2r$,
 that is $D = D_0 \cup D_1 \cup D_2$ and by Fact~\ref{f:close_int} $\cI_0 = \cN/(2(2r)^\alpha)$.
 We also introduce four events that, when holds at the same time, allows for every station $v\in D_0$ to hear a message.
 Note that a station hears a message transmitted by itself.

\begin{enumerate}
\item[(1)]
$E_1$ - exactly one station from $D_0$ transmits
\item[(2)]
$E_2$ - no station from $D_1$ transmits
\item[(3)]
$E_3$ - interference from $D_2$ is lower than $\cI_0/2$
\item[(4)]
$E_4$ - interference from stations outside $D$ is at most $\cI_0/2$
\end{enumerate}

Observe that, since events are independent, we have 
$$Pr(\text{every }v\in D_0\text{ hears something}) \ge Pr(E_1)Pr(E_2)Pr(E_3)Pr(E_4).$$

%Here we give only statement of claims that, when combined, give us the desired result. Details and proofs of this are given in the full
%version of the paper.

\begin{claim}
 $Pr(E_1) \ge q/2$
\end{claim}

\begin{proof}
 Observe that $1/2\ge \sum_{w\in D_0}p_w\ge q$, thus by Fact~\ref{f:sum} we have $Pr(E_1)\ge q/2$.
\end{proof}

\begin{claim}
 $Pr(E_2) \ge (1/4)^{a^\gamma z^\gamma q}$
\end{claim}

\begin{proof}
$D_1$ can be covered by $\chi(r, ar)\le a^\gamma$ balls with radius r,
and by Lemma~\ref{l:playoff2}.2 in every such ball the mass of probability is at most $z^\gamma q$,
thus $\sum_{w\in D_1}p_w \le a^\gamma z^\gamma q$.
Using the inequality from Fact~\ref{f:notransmit} we get $Pr(E_2) \ge (1/4)^{\sum_{w\in D_1}p_w} \ge (1/4)^{a^\gamma z^\gamma q}$.
\end{proof}

We choose such $b$ that we can accommodate all the probability from an unit ball into $D_2$ without violating the condition (2)
Lemma~\ref{l:playoff2}. This allows us for bounding interference in two stages. The first stage is bounding interference
from close stations, and we do it more carefully as the close stations can introduce large noise. Then we bound the interference
from far stations by using the fact that in every unit ball the mass of probability is at most $C_1$.
By the properties of the metric $D_2\cup D_1$ can be covered by $b^\gamma$ balls of radius $r$,
so we choose $b$ as least integer that satisfy $b^\gamma z^\gamma q \ge C_1$. We also set $a=2$ and $z=6$ as a result of the previous observations.

\begin{claim}
\label{c:interf}
 $Pr(E_3) \ge 1/2$
\end{claim}

\begin{proof}
Let $\cI(D_2)$ be the interference generated by the stations from $D_2$.
We split $D_2$ into \emph{layers} $B(x,(a+i+1)r)\setminus B(x,(a+i)r)$. %, for $i=0\dots b-a-1$.
Each of them can be covered by $O(i^{\gamma-1})$ balls of radius $r$.
For the sake of clarity we assume in our calculations that the constant hidden behind the $O$ is $1$.

By bounding the interference from each layer inside $D_2$ we get
$$E[\cI(D_2)] \le \sum_{i=a-1}^{b-1}i^{\gamma-1}z^\gamma q P/(ir)^\alpha = z^\gamma q P/r^\alpha \sum_{i=a-1}^{b-1}i^{\gamma-\alpha-1}\le z^\gamma q P/r^\alpha \sum_{i=a-1}^{\infty}i^{\gamma-\alpha-1}$$.

Observe that, since $P=\beta\cN$, by choosing $q\le 1/(z^\gamma \beta 2^{\alpha+2}\sum_{i=a-1}^{\infty}i^{\gamma-\alpha-1})$
we bound the expected interference from $D_2$ by $\cI_0/4$. Note that the sum $\sum_{i=a-1}^{\infty}i^{\gamma-\alpha-1}$ corresponds to the Riemann zeta function $\zeta(s)$ for $s=\alpha-\gamma+1>1$, which 
converges to a real value. Thus, by Markov's Inequality we get $Pr(E_3)\ge 1/2$.
\end{proof}

\begin{claim}
 $Pr(E_4) \ge 1/2$
\end{claim}
\begin{proof}
We bounded the interference from stations in $B(x, 1)$ in the previous claim. Now we show that the expected outer interference
is as well not too big.
Let $D^c$ denote the set of all stations outside of $D$, $D^c = V\setminus D$.
Since in every unit ball the mass of probability is at most $C_1$ we can bound expected value of $\cI(D^c)$ in the similar fashion as
in the Claim~\ref{c:interf}. By splitting the space into a layers $B(x,i+1)\setminus B(x,i)$ and covering each layer with $O(i^{\gamma-1})$
unit balls, where the probability is at most $C_1 c_\eps$, we get
$$E[\cI(D^c)] \le \sum_{i\ge 1} i^{\gamma-1} C_1 c_\eps P /(i^\alpha) \le C_1 c_\eps P\sum_{i\ge1} i^{\gamma-\alpha-1} \le \cI_0 / 4$$
where the last inequality follows from the fact that $r\le (\frac{2q}{c_\eps C_1})^{1/\gamma}$ which implies $\cI_0 = \cN/(2(2r)^\alpha) \ge \cN c_\eps C_1/(2^{\alpha+4}q)$, and the choice of
$q \le 1/(2^{\alpha+4}\beta \sum_{i\ge 1} i^{\gamma-\alpha-1})$.
Again, we use the Markov's Inequality to conclude the proof of the Claim.
\end{proof}

Finally we set $q=1/(z^\gamma 2^{\alpha+4}\beta \sum_{i\ge 1} i^{\gamma-\alpha-1})$ with respect to the bounds from previous two claims.
By combining all the claims we get $p(v) \ge Pr($every $w\in D_0$ hears a message$) \ge q/8\cdot (1/4)^{a^\gamma z^\gamma q} = 2c_3/c_2$.
\end{proof}

\comment{
\begin{fact}
  For every set $D_0$, defined and considered under conditions from Lemma~\ref{l:playoff2}, satisfying \ref{l:playoff2}.1, \ref{l:playoff2}.2, \ref{l:playoff2}.3
  for every $v\in D_0$ DensityTest($v$) and PlayOff($v$) are satisfied, thus $v$ quits in the considered stage whp.
\end{fact}

\begin{proof}
  Since $D_0$ satisfy condition \ref{l:playoff2}.3, by Proposition~\ref{prop:density} we have that DensityTest($v$) returns True.
  From Lemma~\ref{l:receiving} we get that the probability of receiving a message $p(v)\ge2c_3/c_2$. Using the Chernoff Bound we get that the probability
  of receiving less than $c_3\log n$ out of $c_2\log n$ messages is negligible, thus PlayOff returns True whp.
\end{proof}

This concludes the proof of Lemma~\ref{l:playoff}.
}

Given the result of Lemma~\ref{l:playoff2}, we are ready to finish the proof
of Lemma~\ref{l:playoff}. Choose $c_2$ and $c_3$ such that $2c_3/c_2=q/8\cdot (1/4)^{a^\gamma z^\gamma q}$.
Lemma~\ref{l:playoff2} guarantees that there exists a ball $D_0$ satisfying conditions
(1)--(3) from this lemma. Then, Lemma~\ref{l:receiving} guarantees that 
Playoff returns true whp, and (3) that DensityTest returns true whp for
each active element of $D_0$. This concludes the proof of Lemma~\ref{l:playoff}.
}

\def\lemmatwo{
\subsection{Proof of Lemma~\ref{lemma2}}
\labell{sec:lemma2}
%\subsection{Proof of Lemma~\ref{lemma2}}\labell{sec:lemma2}

Let $\eps'=\eps/2$.
Here we prove %the Second Probability Stabilization Lemma (
Lemma~\ref{lemma2}
%)
stating that for every station $v$ 
the probability distributed among the stations from $B(v,\eps')$
after execution of StabilizeProbability procedure is bounded from
below by a constant.
We show this for given $C_1, c_2, c_3$, we also fix the values of $c_\eps$ and $C_2$ to be respectively
%and %$x\cN\beta/(\eps^\alpha C_1 \theta)$, where $x= 8 ln(\frac{4c_2}{c_3})$ and $\theta$ is some small constan specified later.
$c_\eps\gets 1/(\eps^\alpha C_1 c_d)\cdot 8\ln(4c_2/c_3)$ and
$C_2\gets \min(c_3/(8c_2), C_1c_d/(2c_d))/c_\eps$,
where $c_d$ is the constant from Prop.~\ref{prop:density}.

Intuitively, $c_\eps$ is chosen large enough to make successful transmissions on distance
larger than $\eps'$ very unlikely, provided the probability mass in a considered unit ball
is large enough (close to $C_1$). Increased probabilities generate large noise preventing
communication on distance larger than $\eps'$. 
Since the Playoff uses a constant probability scaled up by the factor of $c_\eps$ in the 
following proof we use the constant $C_2'=C_2 c_\eps$ for the sake of clarity.
%The lemma will be a consequence of the following few observations about a conditions necessary for a single station to fix its probability.
%A station can set its probability
%by quitting from execution of the algorithm.

The main step in the proof is to show that, by our choice of $C_2$ and $c_\eps$,
a station whp does not quit if the mass of probability in its close proximity is small.

\begin{lemma}
\labell{lemma:lowpbb}
 For every node $v$ if $$\sum_{\substack{w\text{ is active}\\ w \in B(v,\eps')}} p_w c_\eps< 2C_2' $$
 then the probability that $v$ turns off with color $p_v$ is negligible.
\end{lemma}

First, we show that Lemma~\ref{lemma2} follows from Lemma~\ref{lemma:lowpbb}.
Then, the proof of Lemma~\ref{lemma:lowpbb} will be provided.
Let one execution of lines 
4-7 in StabilizeProbability be a {\em phase} of the algorithm.
Note that all stations switched off (by quitting in line 7) during a phase have the same color (probability).
Consider any $v$ participating in an execution of the protocol. If any $w\in B(v,\eps')$
does not switch
off until the end of the last phase then the final value of $p_w$ is equal to $2p_{\max}=2C_2'/c_\eps$ and
therefore $\sum_{\substack{w: p_w=2p_{\max} \\ w \in B(v,\eps')}}p_w c_\eps \geq C_2'/c_\eps$
which in turn means that
the statement of Lemma~\ref{lemma2} is satisfied.
Thus, consider the case that all elements of $B(v,\eps')$ switch off before the last phase.
Let $j$ be the phase in which $v$ quits and let $B_i$ denote the set of stations from
$B(v,\eps')$ active after the $i$th phase. Then, by Lemma~\ref{lemma:lowpbb},
$\sum_{w\in B_{j-1}}p_w c_\eps\geq 2C_2'$ whp. Let $k$ be the last phase such that
$\sum_{w\in B_{k}}p_w c_\eps\geq 2C_2'$. Such phase $k$ exists because of our assumption that all
stations eventually switch off. As the probabilities of active stations are multiplied
by $2$ after each phase, the condition $2\sum_{w\in B_{k+1}}p_w c_\eps< 2C_2'$ due to the choice
of $k$. Thus, the set $B_{k}\setminus B_{k+1}$ of stations quitting in the $(k+1)$st phase
satisfies $\sum_{w\in B_k\setminus B_{k+1}}p'_w c_\eps\geq C_2'$, where $p'_w=2p_w$ is the probability
assigned to $w$ in phase $k+1$. This shows that the color assigned to stations in
phase $k+1$ satisfies the inequality from Lemma~\ref{lemma2}.

\comment{ % wersja Michala
The correctness of Lemma~\ref{lemma1} follows from Lemma~\ref{lemma:lowpbb}.
The station $v$ (whp) turns off only when the probability mass in $B(v,\eps')$ is at least $2C_2'$.
If there is some station $w\in B(v,\eps')$ alive at the end of the phase, then $w$ quits with probability $p_{max}=C_2'/c_\eps$.
Otherwise the sum of probabilities of active stations from $B(v,\eps')$ at the end is $0$.
By Lemma~\ref{lemma:lowpbb} there was a round when the sum exceeded $2C_2'$. If in every stage stations with cumulative probability at most $C_2'$ would have been turned off then, by the fact that every station
that survives PlayOff doubles its probability, the sum of probability over active nodes inside $B(v,\eps')$ would never reach $0$. So there have been a stage that the set of stations $T$ have quit with the same color, and
$\sum_{w\in T}p_w c_\eps > C_2'$. This concludes our observation.
}

Now, it remains to prove Lemma~\ref{lemma:lowpbb}. Observe that $v$ turns off only when it receives at least $c_3\log n$ messages during the PlayOff and
DensityTest($v$) returns True. By Proposition~\ref{prop:density} we only need to consider the case when

\begin{enumerate}
\item[(a)]
$\sum_{w \in B(v,\eps')} p_w c_\eps< 2C_2'$ and
\item[(b)]
$\sum_{w \in B(v,1)} p_w \ge C_1 c_d$.
\end{enumerate}
In the remaining part of the proof we show that the probability of receiving $c_3\log n$ messages during PlayOff is negligible for our choice of $C_2'$ and $c_\eps$, provided (a) and (b) hold.

Let $p(v)$ denote the probability that $v$ receives a message if all active stations transmit
with currently assigned probabilities. Below, we express a condition regarding $p(v)$ which is sufficient for correctness of Lemma~\ref{lemma2}.

\begin{fact}\labell{f:cher:lem2}
 If $p(v) < c_3/(2c_2)$ then PlayOff($v$) returns False with high probability.
\end{fact}

\begin{proof}
 The expected number of rounds in which $v$ receives a message is $c_2\log n \cdot p(v) \le \frac{c_3\log n}{2}$. Thus, by
 Chernoff Bound we can make the probability that PlayOff($v$)=true ($v$ receives at least $c_3\log n$ messages) arbitrarily small by increasing $c_2$ and $c_3$ without changing the initial
 ratio $c_3/c_2$.
\end{proof}

In the following, let $S_v$ be the sum of all signals received at node $v$, i.e.,
$$S_v=\sum_{w\text{ is transmitting}} \beta\cN/(\dist(v,w))^\alpha.$$

The next two facts show the way to bound the probability $p(v)$ of receiving a message by $v$
in terms of the sum of probabilities in close neighborhood of $v$ and of the interference from the
whole network. These facts combined with Fact~\ref{f:cher:lem2} give the property claimed
in Lemma~\ref{lemma2}.
%with respect to the interference.

\begin{fact}
\label{fact:maxint}
If $S_v > \frac{2\beta\cN}{\eps'^\alpha}$ then $v$ cannot receive a message from the outside of $B(v,\eps')$.
\end{fact}

\begin{proof}
The strength of signal from node in a distance at least $\eps'$ from $v$ is at most $\beta\cN/\eps'^\alpha$ (since $P=\beta \cN$),
and the interference $\cI$ is at least $S_v-\frac{\beta\cN}{\eps'^\alpha}>\frac{\beta\cN}{\eps'^\alpha}$.
From these observations we have $SINR \le \beta\cN/(\eps'^\alpha(\cN + \cI)) < \beta/(\eps'^\alpha + \beta) < \beta$, since $\beta \ge 1$.
\end{proof}

\begin{fact}
\label{fact:decomp}
The probability of receiving a message at $v$ can be bounded as follows
  $$p(v) < Pr\left( S_v \le \frac{2\cN\beta}{\eps^{\alpha}}\right) + \sum_{w\in B(v,\eps')}p_w\cdot c_\eps$$
\end{fact}

\begin{proof}
The former summand corresponds to the event of receiving a message from a station in a distance at least $\eps$, thus by the Fact~\ref{fact:maxint},
the sum of signals at $v$ should be at most $\frac{2\cN\beta}{\eps^{\alpha}}$. The latter is trivial upper bound on the probability of receiving a message from some
station in $B(v,\eps')$.
\end{proof}

Note that in the setting considered in the proof of Lemma~\ref{lemma:lowpbb} the quantity
$\sum_{w\in B(v,\eps')}p_w\cdot c_\eps$ can be bounded by $2 C_2'$.
By our choice of $C_2'\le c_3/(8c_2)$, in order to prove Lemma~\ref{lemma:lowpbb}, it suffices to show that $Pr(S_v \le \frac{2\cN\beta}{\eps^{\alpha}}) < c_3/(4c_2)$.
Then, by Fact~\ref{fact:decomp}, if $\sum_{w\in B(v,\eps')} p_w c_\eps< 2C_2 = c_3/(4c_2)$ then $p(v) < c_3/(2c_2)$ which in turn gives the statement of Lemma~\ref{lemma:lowpbb} (by Fact~\ref{f:cher:lem2}).

From now on we focus on bounding the probability that interference at $v$ allows for successful transmission on the distance $\eps$ or greater, i.e. that $S_v \le \frac{2\cN\beta}{\eps^{\alpha}}$. As we already pointed out, it is sufficient to bound this probability from above by $c_3/(4c_2)$.

\begin{proposition}
\label{prop:ceps}
 Assuming that DensityTest($v$) is satisfied the following inequality holds with high probability:
 $Pr(S_v \le \frac{2\cN\beta}{\eps^{\alpha}}) < c_3/(4c_2)$.
\end{proposition}
\begin{proof}
In the proof, we take advantage of the fact that $S_v$ is at least the sum
of signals arriving from $B(v,1)\setminus B(v,\eps')$ and, according to (a), (b) and the fact
that $C_2'\leq C_1 c_d/2$, the sum of probabilities in this area is at
least $C_1c_d/2$.
As each transmitter in Playoff uses its probability scaled up by the factor $c_\eps$,
$$\begin{array}{rcl}
E(S_v)&\geq& c_\eps\sum_{w\in B(v,1)\setminus B(v,\eps')}p_w\cdot \beta\cN/(\dist(w,v)^\alpha)\\
&\geq&
c_\eps\sum_{w\in B(v,1)\setminus B(v,\eps')}p_w\cdot \beta\cN\\
&\geq& %/(\dist(w,v)^\alpha)
\beta\cN c_\eps C_1 c_d/2\\
&=&\frac{2\beta\cN}{\eps^\alpha}\cdot 4\ln(4c_2/c_3)
\end{array}$$
where the second inequality follows from the fact that $\dist(v,w)\leq 1$ for
each $w\in B(v,1)$, the third inequality from the fact that $\sum_{w\in B(v,1)\setminus B(v,\eps')}p_w\geq C_1 c_d/2$.

On the other hand, the above estimation of $S_v$ can be seen as the sum
of independent random variables $X_w$ equal either $0$ or $\beta\cN /\dist(v,w)^\alpha$
over all $w\in B(v,1)\setminus B(v,\eps')=B$.
Thus, each of these variables satisfies $X_w\leq \beta\cN/\eps^\alpha$
and $E(S_v)\geq E(\sum_{w\in B}X_w)\geq \frac{2\beta\cN}{\eps^\alpha}\cdot 4\ln(4c_2/c_3)$.

Now,we scale the variables $X_w$ in order to apply the Chernoff bound.
Let $Y_w=X_w/(\beta\cN/\eps^\alpha)\leq 1$, let $Y=\sum_{w\in B}Y_w$.
Then,
$$E(Y)=E(\sum_{w\in B} X_w)/(\beta\cN/\eps^\alpha)\geq 2\cdot 4\ln(4c_2/c_3).$$
Moreover,
$$\begin{array}{lclclcl}
Pr(S_v\leq \frac{2\cN\beta}{\eps^\alpha})&<&Pr(\sum_{w\in B} X_w\leq \frac{2\cN\beta}{\eps^\alpha})\\
&=&Pr(Y\leq 2) \\
&\leq& Pr(Y\leq \frac12 E(Y))\\
&\leq& \exp(-E(Y)/8)\leq c_3/(4c_2)
\end{array}$$
where the third last inequality follows from $4\ln(\frac{4c_2}{c_3})\geq 2$,
since $c_2\geq c_3$.
This finishes the proof of Proposition~\ref{prop:ceps}.
\end{proof}

}

\lemmaone

\lemmatwo

%\input{lemma1.tex}

%\input{lemma2.tex}

%\begin{observation}
% Let $I_v$ be a sum of signals received by $v$ in a given round, and $p(v)$ denote the probability of receiving message by $v$.
% If $I_v > \frac{2\beta N}{\epsilon^{\alpha}}$ then $$p(v) \le \sum_{w\in C(v,\epsilon)} p_w.$$
%\end{observation} 
\section{Broadcast}
\label{s:broadcast}

%\vspace*{
\subsection{Broadcast with non-spontaneous wakeup}

\paragraph{Algorithm NoSBroadcast}
For the model with non-spontaneous wake-up,
we present the algorithm NoSBroadcast in which a message is disseminated over the network
in time $O(D\log^2 n)$. % in an ad hoc setting when stations have no prior knowledge.
The algorithm works in $D$ phases.
Each phase has $O(\log^2 n)$ rounds and consists of two parts.
A node participates in the phase (is {\em active}) if it knows the source message at the beginning of the phase.
The first part of a phase executes StabilizeProbability on the set of active stations.
This execution takes $O(\log^2 n)$ rounds.
As a result, it assigns a color $p_v$ to each active node $v$.
This coloring satisfies conditions from Lemma~\ref{lemma1} and Lemma~\ref{lemma2}.
In the second part,
%step,
each active node transmits the message with probability
$\frac{p_v}{c\eps\log n}$, for some constant $c$, for $O(\log^2 n)$ rounds.
%The first step is quite complicated to be analyzed and we describe it in the
%next section of this paper.
Consider any shortest path $s=v_0,v_1,\ldots,v_k$ in the communication graph from the source $s$ to a node $v_k$.
Our construction guarantees that the $i$-th vertex $v_i$ of the path
knows the source message after the $i$-th phase of the algorithm whp.

\begin{theorem}\labell{th:broad:no}
The NoSBroadcast algorithm solves the broadcast problem in the
non-spontaneous
%non-synchronous
wakeup model
in $O(D\log^2n)$ rounds whp.
\end{theorem}

Theorem~\ref{th:broad:no} follows directly from Fact~\ref{f:SP-time} and the
%following
lemma below.

\begin{lemma}\labell{lem:second:part}
There exists a constant $c$ such that  each neighbour in the communication
graph of each active node $v$ receives the source message whp in the second part of a phase.
%
%node informs all its neighbours in $G$
%about the message whp.
\end{lemma}

%\gs{The proof of the above lemma is based on Lemmas~\ref{lemma1} and \ref{lemma2},
%c.f., Appendix~\ref{sec:lemma1} and~\ref{sec:lemma2} for the proofs.}
%which are proven in the next section of the paper.}
Let us recall that $\eps'=\eps/2$.
In order to satisfy the claim of Lemma~\ref{lem:second:part}, it is sufficient that,
for each active node $v$, an active station $v'$ such that $\dist(v,v')\leq \eps'$
transmits and is heard in distance $1-\eps'$ during NoSBroadcast (see Fact~\ref{f:onbehalf}).
We show that this is actually the case in the following proposition, which concludes
the proof of Lemma~\ref{lem:second:part}, and thus also Theorem~\ref{th:broad:no}.
\begin{proposition}
\label{lmain}
There exists a constant $c$ such that for any node $v$ active in a phase there exists
a node
$v'\in B(v,\eps')$
%$v'\in C(v,\eps')$
which transmits in the second
part of the phase
%step,
and $v'$ is heard anywhere in the distance $1-\eps'$ %in this phase 
whp.
\end{proposition}
\begin{proof}
The sufficient condition for occurrence of the event from the
proposition
%lemma
in a given round is that the following three assertions hold:
\begin{enumerate}
\item[(1)] 
%$C(v,\eps')$
exactly one station transmits in
$B(v,\eps')$,
\item[(2)] no other station in
$B(v,2)$
%$C(v,2)$
transmits,
\item[(3)] the interference from outside of
$B(v,2)$
%$C(v,2)$
in any point of
$B(v,1)$
%$C(v,1)$
      is smaller, than
$\cI=\cN\alpha\eps'$
%$I=N\alpha\eps'$
      (which allows hearing transmissions from the distance $1-\eps'$).
\end{enumerate}

\noindent Now, in a given round, we bound from below the probabilities of the events (1)--(3) by choosing sufficiently large $c$.
\begin{enumerate}
\item[(1)]
%\begin{fact}
The probability, that exactly one node in
$B(v,\eps')$
%$C(v,\eps')$
transmits is bigger than
%$\frac{C_2}{2c\eps c_\eps \log n}$
$\frac{C_2}{2c\eps \log n}$
whp.
This follows from Fact~\ref{f:sum}.
%\end{fact}
\item[(2)]
%\begin{fact}
The probability, that 
no one transmits in
$B(v,2)\setminus B(v,\eps')$
%$C(v,2)\setminus C(v,\eps')$
is bigger than $3/4$ whp.

When restricting to stations $w$ of the same color, the inequality $\sum_{w\in B}p_w\leq C_1$ %\cc{to chyba dla ustalonego koloru?} 
holds for any
unit ball $B$, by Lemma~\ref{lemma1}.
Thus, using the bounded growth property, we know that for each color $\sum_{w\in B(v,2)}p_w=O(1)$.
As there are at most $\log n$ colors, this sum over all colors fulfills %$\log n\frac1{c}\sum_{w\in \dk{B(v,2)}}p_w$ %\cc{Zamiast tego wyrazenia nie powinno byc 
\[\sum_{w\in B(v,2)}\frac{p_w}{c\eps\log n}=\log n \frac{O(1)}{c\eps\log n} = O(1).\]
%is the upper bound on the average number of transmitters from \dk{$B(v,2)$} in a round.
Thus, if $c$ is sufficiently large, the average number of transmitters
is smaller than $1/4$, and by Markov bound no one transmits with probability $3/4$.
%\cc{mr: powinno byc with pbb. $1/4$, sprawdzic jaki wplyw na dalsze rachunki}
This holds whp.

%\end{fact}
\item[(3)]
%\begin{fact}
The probability, that in some point of $B(v,1)$ 
the interference exceeds
$\cI=\cN\alpha\eps'$ is smaller than $1/4$ whp.

Once again, it is sufficient to choose $c$ large enough so that the 
expected maximum of interference in $B(v,1)$ from outside of $B(v,2)$ is smaller than $\cI/4=\cN\alpha\eps'/4$.
This follows from interference estimations similar as in Fact \ref{fact7}.
Then this maximum interference is at most $\cI$ with probability $3/4$ by the Markov bound.

%sum
%of transmission probabilities $p_v$ (which is $O(\log n)$, by properties
%of StabilizeProbability) multiplied by $1/c$ is smaller than \dk{$\cI/4$.}
%Then, the expected interference is at most \dk{$\cI/4$} and it is at most \dk{$\cI$}
%with probability $1/4$ by the Markov bound.
%\end{fact}
\end{enumerate}

As the events (1)--(3) are independent, the probability that exactly one node in
$B(v,\eps')$
transmits in a given round and it is heard in range $1-\eps'$ is bigger than
%\[C_2/(2c\eps c_\eps\log n)\cdot\frac34\cdot\left(1-\frac14\right)>C_2/(4c\eps c_\eps\log n)
\[C_2/(2c\eps \log n)\cdot\frac34\cdot\left(1-\frac14\right)>C_2/(4c\eps \log n)\ ;
\]
%\dk{this holds} whp.
%\cc{Poprzednie zdanie niejasne - probability ze zostanie uslyszany ... jest takie i takie whp}
For further references, we state it as a separate fact.

\begin{fact}\labell{f:one:roundc}
There exists a constant $c$ such that if each node $v$ is transmitting with probability
$\frac{p_v}{c\eps\log n}$ in a round,
the probability that exactly one node in $B(v,\eps')$  
transmits and it is heard in range $1-\eps'$ is at least %$p=C_2/(4c\eps c_\eps \log n)$.
$p=C_2/(4c\eps \log n)$.
\end{fact}

%\begin{proof}(of Lemma \ref{lmain})
Now, we take into account that the second part of a phase lasts for many subsequent rounds.
%\dk{repeat it in many subsequent rounds.}
%do many rounds.
Let the number of rounds be $T=(a\ln n)/p=O(\log^2 n)$, where
%$p=C_2/(4c\eps\log n)$ and 
$a$ is an arbitrary constant.
The~probability that not all stations in range $1-\eps$ from a given active node $v$
get the message during the second part is at most
\[(1-p)^{a\ln n/p}<e^{-a\ln n}=n^{-a}
\ .
\]
Hence, all neighbours of $v$ in $G$ get the source message whp during part two of the phase.
\end{proof}

%\newpage
\subsection{Broadcast with spontaneous wakeup}

%\cc{W tej sekcji nie sa jasne stale: $C,c,p$. TJ: chyba ok; $c$ jest wyjasniona, a $c$ i $p$ nie ma}

\paragraph{Algorithm SBroadcast}
%Now we assume, that we already have colors assigned to the nodes e.g. during some first message is %broadcasted.
Now, we consider the model with spontaneous wake-up.
We present an algorithm SBroadcast, which starts from
a {\em single} execution of StabilizeProbability on all stations from
the network for $\eps''=\eps/3$.
%
%We show that the further messages can be broadcasted even faster.
The color assignment made during StabilizeProbability %can play the same
might be viewed as a kind of
%role as 
a communication backbone.\footnote{%
In Yu et al.\ \cite{YuHWTL12} the backbone is obtained as a connected dominating set
of the communication graph, by applying techniques from geometric radio networks.
However, this result holds only for a restricted family of networks and the analysis
requires that stations cannot receive signals from distance larger than $1-\eps$,
i.e., their model assumes weak devices, and it is known that the power of that model is different
than considered in our work (c.f.,~\cite{JKS13i})}
When coloring is done, the source node transmits the message deterministically.
Then,
after receiving the source message, every node transmits it %, each time}
%each other node transmits the source message after receiving it,
%in each round
with probability $\frac{p_v}{c\eps\log n}$ in each of 
next $O(\log^2 n)$ consecutive rounds. (The values of $c$ and $p$ are the same as in Fact~\ref{f:one:roundc}.)
This assures that in each round the broadcast message is propagated by one edge, say $\{v,u\}$, of the communication
graph ie. $B(v,\eps'')$ informs all nodes of $B(u,\eps'')$ with probability $p$. 
Using standard concentration bounds
for sums of independent random variables, one can assure that the source message
is delivered to each station in the graph-distance at most $D$ from the source
in $O(D\log n+\log^2n)$ rounds whp.

\begin{theorem}\labell{th:br:spont}
The SBroadcast algorithm solves the broadcast problem in the
spontaneous
%synchronous
wake-up model
in $O(D\log n+\log^2n)$ rounds whp.
\end{theorem}
%
%Before we prove the Theorem we remind some version of Chernoff bound
%
%
\begin{proof}%(of Theorem~\ref{th:br:spont})
If $B(v,\eps'')$ knowing the message informs any $B(u,\eps'')$ for $\{v,u\}\in G$ whp
which can be proved similarly to Proposition~\ref{lmain}.
Consider a shortest path from $s$ to some node $v$.
A sufficient condition for $v$ to receive a message from $s$ is
$D$ successful transmissions on this path.
In each transmission a subsequent node $v_i$ gets the message together with
its neighbourhood $B(v_i,\eps'')$. 
%This is done in each round with probability $p$.
By Chernoff bound %Lemma \ref{cher1}
all these transmissions happen whp in time $t=(2D+C\log n)/p=O(D\log n+\log^2 n)$.
\end{proof}

\section{Application to Other Problems}
\label{s:applications}

In this section we outline how we can solve some distributed
network problems other than broadcast using algorithms developed in the main part of the paper.
We assume that all stations share a common global clock
(i.e.\ they all have a common counter value assigned to each round).
Each problem can be considered in adhoc setting or with
some preexisting assignment of colors $p_v$ fulfilling
the conditions
% \dk{in}
from 
Lemmas \ref{lemma1} and \ref{lemma2}.
In the latter case the coloring is used as a backbone.

\paragraph{Adhoc wake-up} 
We study the wake-up problem as considered, e.g.,\ in \cite{CGK07}.
Each node in the network either wakes up spontaneously or gets activated by receiving a wake-up signal from another node. 
All active nodes transmit the wake-up message according to a given protocol. 
The running time of the protocol is the number of steps counted from the first spontaneous wake-up until all nodes become activated. 
Wake-up times are decided by an adversary.

In adhoc setting wake-up can be done analogously to the broadcast in which each awake station 
assumes that it has already received the same wake-up message. 
To assure the synchronization between stations, each (spontaneously) awaken station begins an execution of the protocol in the first round whose
number is divisible by $T$, where $T=O(D\log^2 n)$ is the number of rounds of an execution
of the broadcast protocol.
(Recall that we assume global clock.)
All the stations get this message
whp after time $2T=O(D\log^2 n)$ from the first spontaneous wake-up in a network.

\paragraph{Wake-up with established coloring}

Now, we consider the wakeup problem in the setting, where all station have assigned colors (probabilities) $p_v$
satisfying Lemma~\ref{lemma1} and Lemma~\ref{lemma2}.
Our algorithm for this setting works in two {\em phases}.
In the first phase a new coloring $q_v$ satisfying Lemma~\ref{lemma1} and Lemma~\ref{lemma2}
is found for
stations woken up spontaneously, i.e., by the adversary.
For all other stations $v$ we set $q_v=0$.
Then the message is broadcasted using the color 
$p_v+q_v$
for each station $v$.
This procedure has running time $T=O(D\log n+\log^2 n)$.
To assure the synchronization between stations, a station (woken up spontaneously) begins the protocol's execution in the first round whose
number is divisible by $T$ (i.e., a station woken up spontaneously ignores this spontaneous wakeup event until the earliest round
number divisible by $T$ and starts participating in the $q_v$ coloring then, provided it has not
received a message from other stations up to this moment).

\paragraph{Consensus in adhoc setting}

We consider consensus problem defined as follows.
Each station $v$ has some message $m_v$.
In certain moments some stations wake-up spontaneously 
(i.e., stations are chosen and awaken by an adversary).
In the end of the protocol all stations should agree on the same, say lexicographically smallest, message.
We assume, that the set of
possible messages is $\{0,1,2,\ldots,x\}$.
The running time of the protocol is the number of steps counted from the first spontaneous wake-up until all nodes know that the protocol's execution is finished.

Our protocol makes agreement on the lexicographically smallest message using the following strategy.
At the beginning of the protocol, stations woken up spontaneously perform
wake-up in adhoc setting.
In the last execution of StabilizeProbabilities in this wakeup,
they establish some coloring $p_v$.
Then stations that have the first bit of the message equal to $0$
perform wake-up with established coloring $p_v$ (in a limited time, as given above for the
wake-up with established coloring problem)
as they were woken up spontaneously.
This wakeup is successful if and only if the smallest message
has its first bit 0, so after this wakeup is done all stations learn
the one-bit prefix of the lexicographically smallest message.
This procedure can be %generalized a more general routine that is 
iterated
%repeated
$\log x$ times for consecutive bits of messages' binary representations.
%After last iteration all stations know the lexicographically smallest message.
In the $i$-th iteration stations that have the $i$-bit prefix of the $m_v$
equal to the smallest $(i-1)$-bit prefix already known to everybody plus bit 0 appended,
%perform 
initiate wakeup with established coloring.
This way all stations learn
the $i$-th bit of lexicographically smallest message.
Thus the consensus problem can be solved in time $O(D\log n\log x+\log^2 n\log x)$.

\paragraph{Leader election in adhoc setting}
We consider the leader election problem as a task of choosing one station
in the whole network as the leader, assuming all stations start a protocol
at the same moment.

At the beginning, all stations choose IDs from the set $\{1,\ldots,n^3\}$, independently at random,
which guarantees uniqueness of IDs whp.
Then, the stations perform the consensus protocol, as described above, on assigned 
%temporary 
IDs.
This gives a solution in $O(D\log^2 n+\log^3 n)$ rounds.

\bibliographystyle{abbrv}
\bibliography{references}

\end{document}